\documentclass{article} 
\usepackage{iclr2026_conference,times}
\usepackage{float} 
\usepackage{booktabs} 
\usepackage{array}    

\usepackage{amsmath,amsfonts,bm}

\def\eqref#1{equation~\ref{#1}}

\def\1{\bm{1}}

\DeclareMathAlphabet{\mathsfit}{\encodingdefault}{\sfdefault}{m}{sl}
\SetMathAlphabet{\mathsfit}{bold}{\encodingdefault}{\sfdefault}{bx}{n}

\usepackage{hyperref,graphicx,amsmath}
\usepackage{url}
\usepackage{multirow}
\usepackage{booktabs}
\usepackage{booktabs} 
\usepackage{amsmath}  
\usepackage{makecell}
\usepackage{amssymb}
\usepackage{mathtools}
\usepackage{amsthm}
\usepackage{makecell}
\usepackage{multirow}
\usepackage[capitalize,noabbrev]{cleveref}
\theoremstyle{plain}
\newtheorem{theorem}{Theorem}[section]

\newtheorem{lemma}[theorem]{Lemma}

\theoremstyle{definition}

\newtheorem{assumption}[theorem]{Assumption}
\theoremstyle{remark}
\newtheorem{remark}[theorem]{Remark}
\definecolor{citecolor}{HTML}{114083}
\definecolor{linkcolor}{HTML}{ED1C24}
\hypersetup{colorlinks=true,citecolor=citecolor, linkcolor=linkcolor, urlcolor=linkcolor}
\title{MelTok:2D Tokenization for Single-Codebook Audio Compression}

\author{
Jingyi Li$^{1*}$, 
Zhiyuan Zhao$^{1*}$, 
Zhisheng Zhang$^{4}$, 
Yunfei Liu$^{1}$, 
Lijian Lin$^{1}$, 
Ye Zhu$^{1}$, 
Jiahao Wu$^{1,3}$, 
Qiuqiang Kong$^{2}$, 
Yu Li$^{1\dagger}$ \\
$^{1}$ International Digital Economy Academy (IDEA) \\
$^{2}$ Chinese University of Hong Kong \\
$^{3}$ Shenzhen Graduate School, Peking University \\
$^{4}$ Shenzhen International Graduate School, Tsinghua University \\
\texttt{kimberlylee200106@gmail.com} \\
* Equal contribution \quad \dagger\ Corresponding author
}

\begin{document}

\maketitle
\begin{abstract}
Large Audio Language Models (LALMs) have emerged with strong performance across diverse audio understanding tasks and can be further enhanced by neural audio codecs. Transitioning from multi-layer residual vector quantizers to a single-layer quantizer has been shown to facilitate more efficient downstream language models decoding. However, the ability of a single codebook to capture fine-grained acoustic details remains limited, as the frequency-variant nature of 1D tokenizers leads to redundancy. To address this issue, we propose MelTok, a two-dimensional (2D) tokenizer that effectively compresses acoustic details of 44.1 KHz audio into a single codebook. The tokenizer encodes audio into a more compact representation than one-dimensional tokenizers. Furthermore, to recover audio from mel-spectrogram tokens, we propose a token-based vocoder. Both objective and subjective evaluations demonstrate that MelTok achieves quality comparable to multi-codebook codecs and outperforms existing state-of-the-art neural codecs with a single codebook on high-fidelity audio reconstruction. By preserving acoustic details, MelTok offers a strong representation for downstream understanding tasks.
\end{abstract}
\begin{figure}[H]
    \centering
    \includegraphics[width=1.0\linewidth]{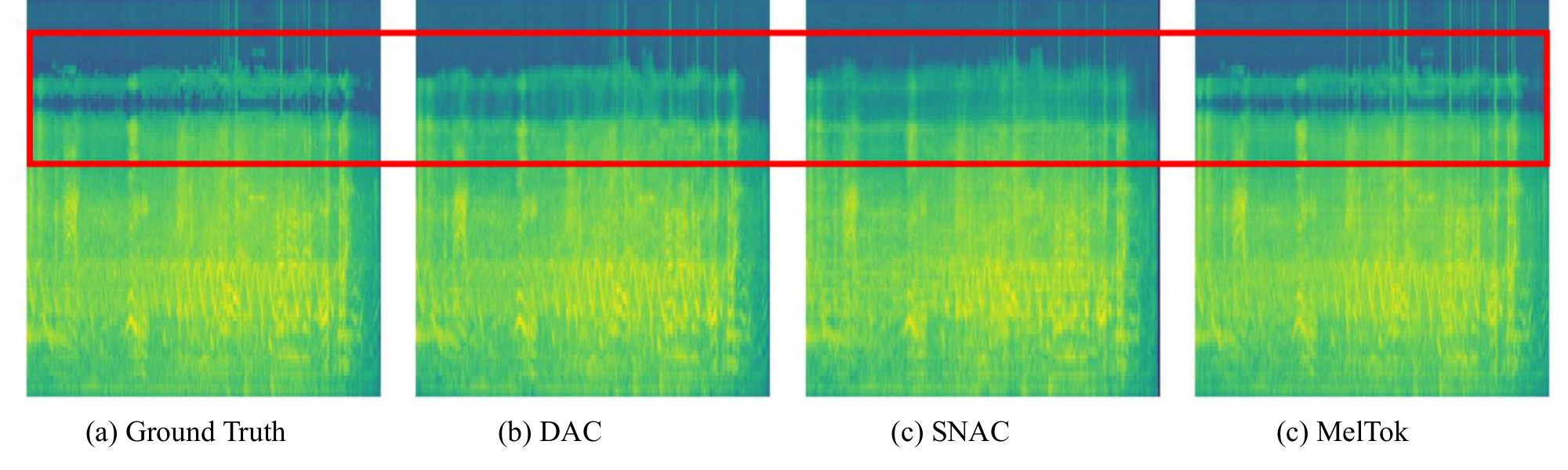}
    \caption{Mel-spectrogram comparison of original and reconstructed waveforms produced by different codecs. MelTok accurately reconstructs the high-frequency details.}
\end{figure}

\section{Introduction}

Discrete audio tokens generated by neural audio codecs compress continuous audio signals into a compact discrete space while preserving perceptual quality and semantic content~\cite{mousavi2025discreteaudiotokenssurvey}. Discrete audio tokens require less storage and achieve faster transmission than continuous embeddings~\cite{theis2017lossyimagecompressioncompressive}. These tokens serve as an efficient and flexible interface for downstream audio understanding tasks such as Automatic Speech Recognition (ASR) ~\cite{radford2022robustspeechrecognitionlargescale} ~\cite{hsu2021hubertselfsupervisedspeechrepresentation}, sound event detection~\cite{7280624}, music genre classification ~\cite{bahuleyan2018musicgenreclassificationusing}, and so on. Audio codecs typically consists of an encoder-quantizer-decoder
structure to encode. The encoder transforms the input waveform into a continuous representation ~\cite{langman2025spectralcodecsimprovingnonautoregressive}, 
The quantizer then maps this continuous representation to a discrete code from a codebook. Finally, the decoder reconstructs the original waveform from the selected code~\cite{agustsson2017softtohardvectorquantizationendtoend}. Compression is achieved when the number of bits used to represent the code is smaller than that required for the original audio signal~\cite{yang2020sourceawareneuralspeechcoding}.

Based on the number of quantizers used, quantization method of codecs can be broadly categorized into two types: multiple stage vector quantization ~\cite{Juang1982MultipleSV} and single vector quantization (SVQ). Full-Bandwidth Audio Codecs, such as DAC ~\cite{kumar2023highfidelityaudiocompressionimproved}, most commonly use residual vector quantizer (RVQ) for quantization. With iterative residual refinement, the multiple stage vector quantizer can compress 44.1 KHz audio with minimal loss in quality~\cite{kumar2023highfidelityaudiocompressionimproved}. However, multi-codebook codecs depend on multi-sequence prediction, which reduces efficiency and robustness ~\cite{li2024singlecodecsinglecodebookspeechcodec}. Single vector quantizers has emerged as simpler and particularly useful tools for downstream generation tasks such as acoustic language models~\cite{ye2025llasascalingtraintimeinferencetime}. Recent work such as WavTokenizer ~\cite{ye2024codecdoesmatterexploring} and UniCodec ~\cite{jiang2025unicodecunifiedaudiocodec} investigates compression in different domains using a single codebook. However, existing single quantizer approaches do not take “full bandwidth” rates of 44.1 or 48 kHz into account. This limitation is critical, as high-resolution acoustic detail is often required for good performance in general downstream tasks.~\cite{valin2016fullbandwidthaudiocodeclow}.

Neural audio codecs can also be categorized into two types based on the representation they compress: waveform-based and spectral-based approaches. In these waveform-based codecs, 1d tokenizers are naturally used to compress one-dimensional time-domain data~\cite{mousavi2025discreteaudiotokenssurvey}.  DAC ~\cite{kumar2023highfidelityaudiocompressionimproved} introduces multiscale mel reconstruction loss into this framework, which better captures details and thus improves audio quality. SNAC ~\cite{siuzdak2024snacmultiscaleneuralaudio} extends Residual Vector Quantization (RVQ) to multiple temporal resolutions, resulting in more efficient compression. In these waveform-based codecs, 1d tokenizers are naturally used to compress one-dimensional time-domain data.

However, high-frequency modeling is still challenging for these waveform-based models especially at high sampling rates with fewer quantizers seen in Figure ~\ref{fig:frequency_detail}, which is incompatible with our single-codebook objective. Spectral-based approaches solve this problem by transforming the waveform into the two-dimensional frequency domain features, which provides a more effecient representation and allows the model to better capture fine-grained frequency details~\cite{guo2025recentadvancesdiscretespeech}. Recent works such as Spectral Codec\cite{langman2025spectralcodecsimprovingnonautoregressive} compress the mel-spectrogram and reconstruct the time-domain audio signal, which typically rely on 1D convolution kernels operating along the temporal axis. 2D convolution kernels are first used in FunCodec~\cite{du2023funcodecfundamentalreproducibleintegrable} after transforming waveform into spectrum, followed by a residual vector quantization (RVQ) module. While most recent codecs adopt a 1D tokenizer, they often suffer from inherent limitations: the resulting representations are not sufficiently compact. And they imperfectly model high-frequency content through one codebook, leading to audio that are clearly distinguishable from originals ~\cite{défossez2022highfidelityneuralaudio}. Our contributions are as follows:
 
\textbf{Compression of 44.1 KHz audio in one single codebook}: This codec aims to encode high-resolution audio using a single quantizer layer with minimal loss of high-frequency content, which is crucial for downstream audio understanding tasks.

\textbf{Perceptual loss in mel-spectrogram reconstruction}: This codec incorporates perceptual loss into mel-spectrogram reconstruction to alleviate the over-smoothing problem, and further relates it to the feature matching loss used in traditional GAN-based codecs.

\textbf{A "frequency-invariance" representation for audio}:
We identify a critical issue in existing single quantizer codecs which can not reconstruct full bandwidth audio due to frequency-variance redundancy. Unlike classical 1D tokenizers, which treat frequency content sequentially and are thus constrained by the assumption of frequency variance, MelTok utilizes a 2D tokenizer. 

\textbf{Token-based vocoder}:We use a two-stage training framework to train 2D tokenizer and token-based Vocoder separately, which leads to better GAN training stability and audio quality.

\begin{figure}
    \centering
    \includegraphics[width=0.8\linewidth]{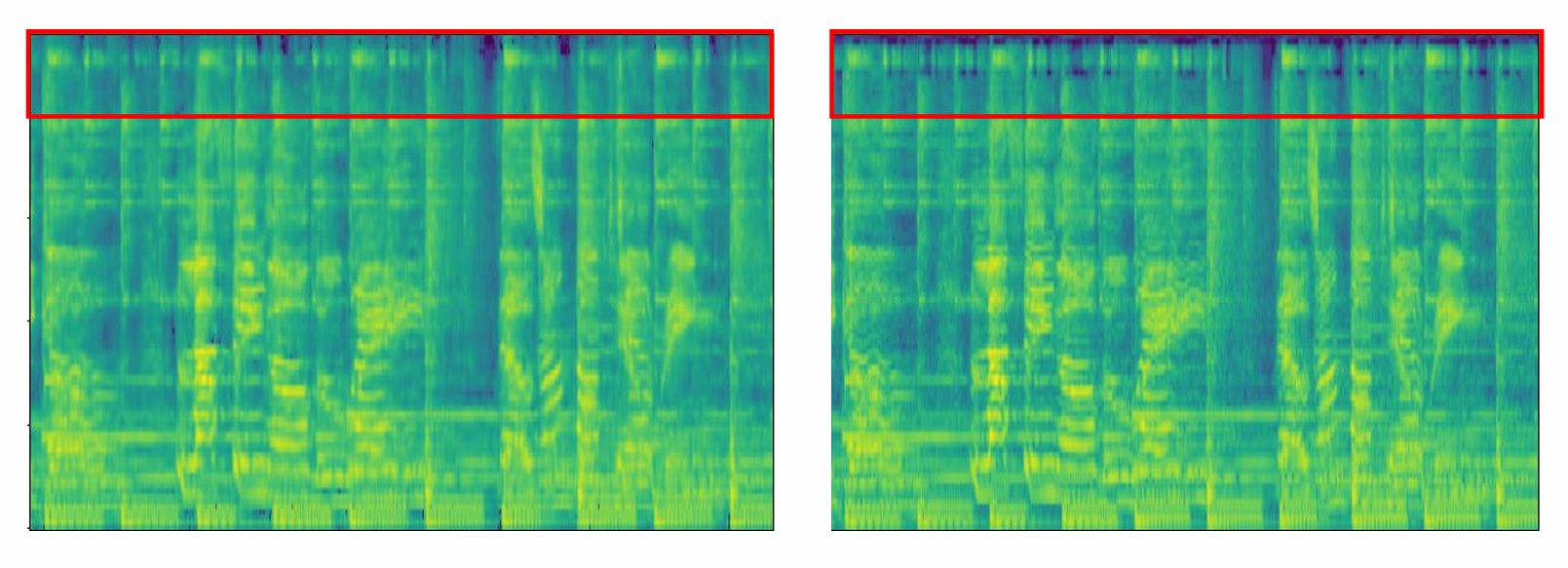}
    \caption{Loss of high-frequency (above 20k Hz) detail in a waveform-based codec. \textbf{Left:} spectrogram of result from a waveform-based codec using 4 quantizers. \textbf{Right:} ground truth (GT). Noticeable differences exist in the high-frequency Mel spectra, resulting in poor reconstruction of high-frequency components,  the bright ringing sound in the original sound.}
    \label{fig:frequency_detail}
\end{figure}

\section{Related Work}
\label{gen_inst}

\subsection{2D Tokenizer}
Existing discrete audio tokenizers compress the waveform or spetrogram along the temporal dimension through Conv1D, so it is defined as a 1D tokenizer. A 2D tokenizer compresses 2 dimentional data in both time and frequency dimensions, and then transforms the compressed representation into a sequence of discrete tokens.  Many works on 2D tokenizers have been explored in the image domain~\cite{yu2024imageworth32tokens}. VQ-VAE ~\cite{yang2020sourceawareneuralspeechcoding} first introduced vector quantization in the latent space of VAEs to map images, audio, and video into discrete values. Vector Quantized Generative Adversarial Network (VQGAN) ~\cite{esser2021tamingtransformershighresolutionimage} extended this by adding perceptual and adversarial losses to better capture detail infermation. ViT-GAN further replaced convolutions with ViT Transformers ~\cite{dosovitskiy2021imageworth16x16words}.  

Some work in audio has also explored the usage of 2D convolutional layers.  For example, PANNs ~\cite{kong2020pannslargescalepretrainedaudio} employ 2D CNNs such as CNN14 to capture time–frequency invariant patterns on Wavegram features, which helps the model improves Audio Pattern Recognition performance. To explicitly capture harmonic structure, a cross-frequency structure, \cite{ulicny2022harmonicconvolutionalnetworksbased} proposed a 2D convolution Harmonic Convolution. Inspired by the success of 2D kernels used in capturing audio structures, we aim to leverage this invariance to derive a more compact discrete representation for high-fidelity neural audio codec. Building on the structure of powerful 2D Transformer tokenizers ~\cite{nvidia2025cosmosworldfoundationmodel}, we fully explore compact 2D representations for audio.

\subsection{Vocoder}
Neural vocoders are neural network models that converts intermediate representations, such as mel-spectrograms, into high-fidelity audio~\cite{jiao2021universalneuralvocodingparallel}. Autoregressive models had long been the best-performing vocoders. WaveNet~\cite{oord2016wavenetgenerativemodelraw}, for instance, uses the mel-spectrogram as a local condition. However, its requirement for sequential (sample-by-sample) generation limits streaming efficiency. GAN-based models are capable of generating speech from mel-spectrogram efficiently~\cite{kong2020hifigangenerativeadversarialnetworks}. Since low latency is a key property for a good codec, we build our model on Vocos~\cite{siuzdak2024vocosclosinggaptimedomain}. Vocos is a fast neural GAN-based vocoder designed to reconstruct audio from mel-spectrogram through inverse Fourier transform.

\section{Methods}
\label{headings}
To achieve faster convergence, we propose a two-stage codec, where the first stage focuses on mel-spectrogram reconstruction with metric losses, and the second stage incorporates a discriminator to recover high-fidelity waveform from mel discrete tokens. This architecture significantly improves training efficiency, enabling our second-stage model to converge within only 50 epochs.

One efficient way to extract spectral features from an audio signal is through the Short-Time Fourier Transform (STFT). Given an input signal $x[n]$ with length T, $X_t[k]$, the STFT coefficient for the k-th frequency bin and the t-th time frame, denoted as $x_t[k]$. To better connect with the human sound perception, the frequency axis of the spectrogram can be mapped onto the Mel scale using a filter bank. This result is known as Mel spectrogram. Finally, the logarithm of the Mel spectrogram is taken to limit the range of values. The log-Mel spectrogram coefficient for the k-th frequency bin and the t-th time frame is given by:
\begin{equation}
\underset{0 \leq m \leq M-1}{LMS_t[m]} 
= \log \left[ \sum_{k=0}^{N-1} H_m[k] \cdot \left| \sum_{n=0}^{N-1} x[n] \, w[n-tH] \, e^{-j 2\pi kn/N} \right|^2 \right].
\label{eq:LMS_full}
\end{equation}
where $H_{m}[k]$ is the $k^{th}$ coefficient for the $m^{th}$ filter bank~\cite{2001SpokenLP}, $w[n]$ is the window function (e.g., Hamming window), H is the hop size, and N is the total number of frequency bins.

The total number of filter banks, denoted as M, determines the frequency resolution of the resulting Mel spectrogram. To preserve high-frequency details, the number of Mel filter banks M should be chosen sufficiently large, i.e., not less than 96.

\subsection{2D Tokenizer VS 1D Tokenizer}

\begin{figure}
    \centering
    \includegraphics[width=1.0\linewidth]{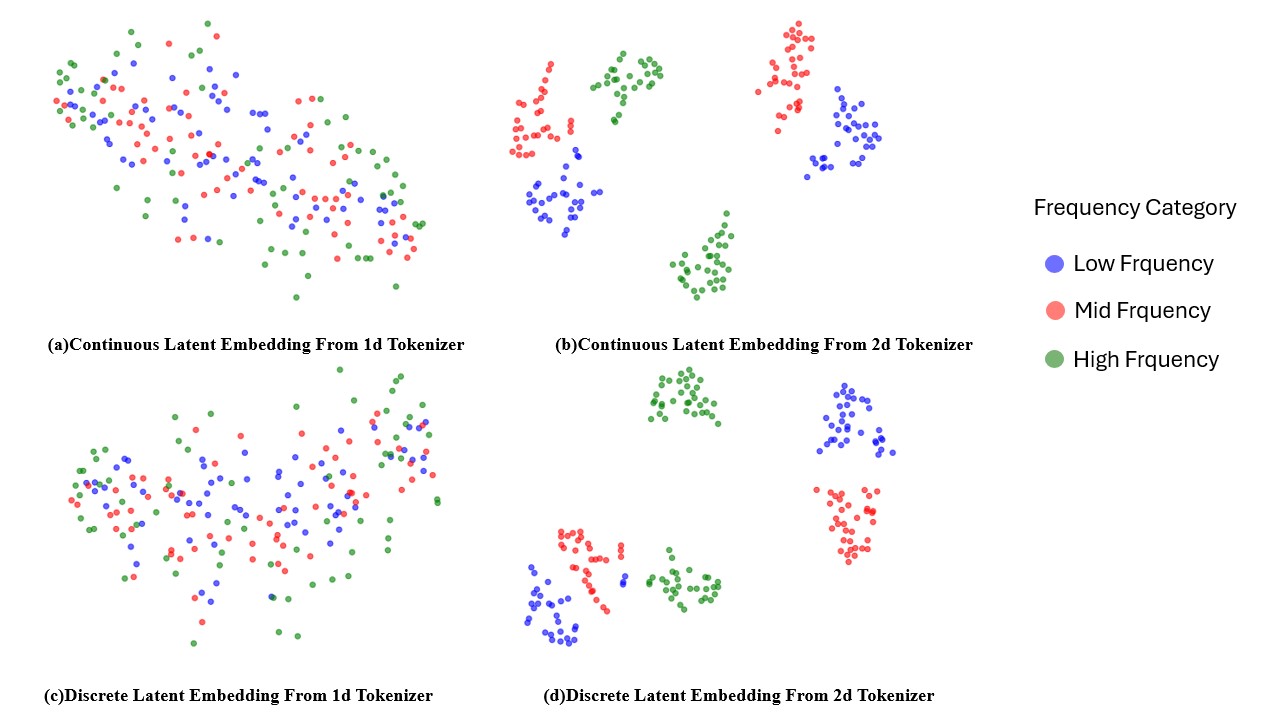}
    \caption{t-SNE visualization of latent embeddings produced by the 1D and 2D tokenizers. Each point represents the latent embedding of a single time frame, and each color denotes a different frequency band.}
    \label{fig:1d2dcompare}
\end{figure}

\begin{figure}
    \centering
    \includegraphics[width=1\linewidth]{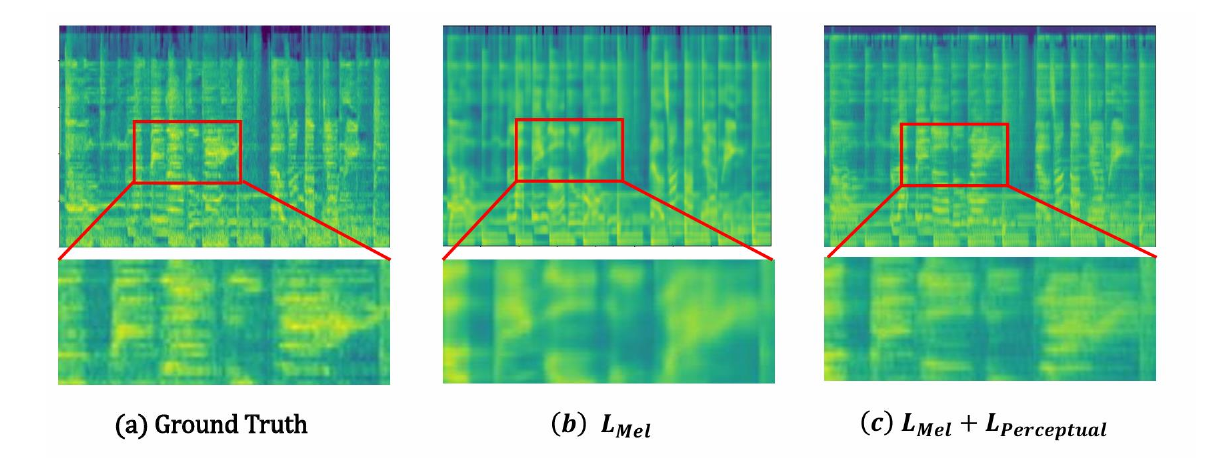}
    \caption{Comparison of reconstructed log-mel spectrograms  trained with different loss. The bottom row shows a zoomed-in view, highlighting the differences in smoothness and spectral sharpness. }
    \label{fig:mel_smoothness}
\end{figure}
Compared with conventional 1D tokenizers, our proposed 2D tokenizer provides a more compact representation for single-codebook audio compression. A 1D tokenizer processes only the temporal dimension, forcing the model to encode all frequency variation into a single token sequence. In contrast, our 2D tokenizer maintains the inherent time–frequency structure, decreasing redundancy caused by frequency variation. Furthermore, patchifying the mel-spectrogram enables the two-dimensional tokenizer to generate frequency-specific embeddings. The detail of patchify can be seen in ~\ref{app:architecture}. This leads to disentangled representations in which high-, mid-, and low-frequency structures are clearly separated in the latent space, seen in ~\ref{fig:1d2dcompare}. 
These findings suggest that 2D tokenization provides a more suitable representation than 1D tokenization for single-codebook audio codecs. The encoder details are given in ~\ref{Table:encoder_detail}.

\begin{figure}
    \centering
    \includegraphics[width=0.8\linewidth]{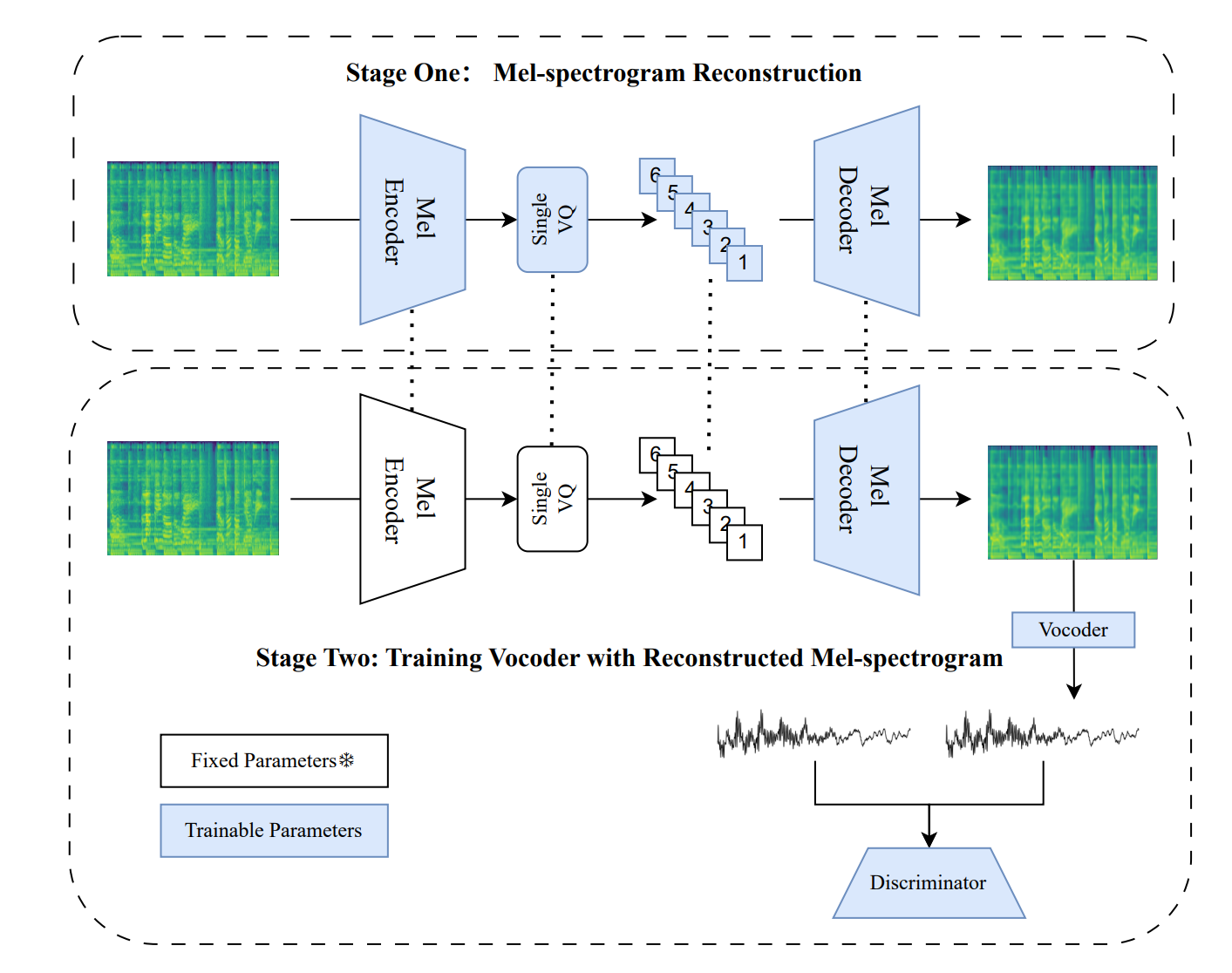}
    \caption{Training paradigm of MelTok.}
    \label{fig:melcap_training}
\end{figure}

\subsection{First Stage: Mel-Spectralgram Reconstruction via 2D Tokenizer}
In the first stage, we compress audio into discrete tokens and then reconstruct the mel-spectrogram from these tokens. For convenience, we adopt the log-mel representation mentioned in ~\eqref{eq:LMS_full}, which helps constrain the value range.  Our method builds on the Cosmos tokenizer ~\cite{nvidia2025cosmosworldfoundationmodel} as the foundational encoder–decoder. We optimize with the L1 loss applied on the log Mel-spectrograms, which minimizes the element-wise difference between the input and reconstructed spectrograms:
\begin{equation}
\mathcal{L}_{\text{Mel}} = \frac{1}{T \cdot M} \sum_{t=0}^{T-1} \sum_{m=0}^{M-1} \left\lVert \mathbf{S}_{t,m} - \hat{\mathbf{S}}_{t,m} \right\rVert_{1},
\end{equation}
where $\mathbf{S}_{t,m}$ and $\hat{\mathbf{S}}_{t,m}$ denote the value of the input and reconstructed Mel-spectrogram, respectively, at the $t$-th time frame and the $m$-th Mel frequency bin.
Using only a reconstruction loss can lead to overly smooth reconstructed Mel-spectrograms, as shown in Figure~\ref{fig:mel_smoothness}. This oversmoothness negatively affects downstream generation tasks such as TTS, causing the synthesized waveform to sound muffled and unnatural~\cite{sheng2018reducingoversmoothnessspeechsynthesis}.  In order to obtain a more detailed Mel spectrogram, we employ perceptual loss based on the VGG-19 features, given by ~\cite{simonyan2015deepconvolutionalnetworkslargescale}. We provide a theoretical justification for the perceptual loss on mel-spectrograms and the feature matching loss used when training generator of the vocoder in the appendix ~\ref{equivalence}.

\begin{equation}
\mathcal{L}_{\text{Perceptual}} = \frac{1}{L} \sum_{l=1}^{L} \sum_{t} \alpha_l 
\left\lVert \mathrm{VGG}_l(\hat{S}) - \mathrm{VGG}_l(S) \right\rVert_{1},
\end{equation}
where $\mathrm{VGG}_{l}(\cdot) \in \mathbb{R}^{T \times M \times C}$ 
denotes the feature maps extracted from the $l$-th layer of a pre-trained VGG-19 network, $L$ is the number of layers considered, and $\alpha_{l}$ is the weight assigned to the $l$-th layer.

\subsection{Second Stage: From Mel-Spectrogram Tokens to Waveform}

Neural vocoders are primarily designed to recover audio waveforms from mel-spectrogram representations~\cite{siuzdak2024vocosclosinggaptimedomain}. In contrast, our goal is to reconstruct audio waveforms from the mel-spectrogram discrete tokens obtained in the first stage, rather than from the ground-truth mel-spectrograms. Usage of these codes as input to the vocoder has advantage: it help stabilize GAN training in the second stage. Given that the first stage employs a VQ-VAE, the resulting mel-spectrograms contain reconstruction errors. As we have theoretically shown, the mapping from a continuous high-dimensional mel-spectrogram to a finite discrete codebook introduces a upper bound on the propagated error. Because the codes are discrete and belong to a finite codebook, the propagated errors are strictly bounded, preventing extreme deviations and ensuring more robust waveform recovery. By contrast, mel-spectrograms live in a continuous space, where errors cannot be strictly bounded, making the waveform recovery more sensitive to small perturbations.

\subsubsection{Analysis: Bounded Error of Discrete Codes}
\label{analyse}

\begin{assumption}[Discrete Code Quantization]
\label{assumption}
Let $\mathbf{s}$ denote the original mel-spectrogram and $\mathbf{c}  \in \mathcal{C}$ be the discrete token obtained from a VQ-VAE encoder $E$, where $\mathcal{C}$ is a finite codebook. 
We assume that the quantization error due to mapping $\mathbf{s}$ to any code $\mathbf{c}$ in codebook is bounded:
\[
\|\mathbf{c_n} - \mathbf{s}\| \le \|\mathbf{c} - \mathbf{s}\| \le\|\mathbf{c_f} - \mathbf{s}\| = \Delta,
\]
where $c_f$ denotes the farthest code, and $c_n$ denotes the nearest code, $\Delta$ depends on the size and codebook dimension.

\end{assumption}
Let $\mathbf{w}$ denote the waveform reconstructed in the second stage via a neural vocoder $f$. Then the reconstruction is
\[
\mathbf{w} = f(\mathbf{c}).
\]

\begin{lemma}[Lipschitz Bound]
\label{lem:lipschitz}
If $f$ is locally Lipschitz continuous with constant $L$, then
\[
\|f(\mathbf{c}_1)-f(\mathbf{c}_2)\| \le L \|\mathbf{c}_1 - \mathbf{c}_2\|,
\]
\end{lemma}
\begin{theorem}[Bounded Waveform Error]
\label{thm:bounded_error}
Combining Lemma~\ref{assumption} and Lemma~\ref{lem:lipschitz}, the error in the reconstructed waveform due to discrete code quantization is bounded:
\[
\|\mathbf{w} - f(\mathbf{s})\| = \|f(\mathbf{c}) - f(\mathbf{s})\| 
\le L \|\mathbf{c} - \mathbf{s}\| 
\le L \Delta.
\]
Thus, the propagated error from the first-stage discrete token to the final waveform reconstruction is strictly bounded.
\end{theorem}

Assuming that the neural vocoder $f$ is $L$-Lipschitz continuous,
the error propagated to the waveform $\mathbf{w}$ can be bounded by 
$L \Delta$. (The vocoder is a composition of convolution network, activation function and ISTFT; a proof that ISTFT is Lipschitz is given in Appendix~\ref{app:istft_lipschitz}.)

\subsubsection{Model Architecture}
The theoretical analyse provides guidance for vocoder architecture design. For example, in the generator we choose to use the Snake activation function instead of Leaky ReLU. The Snake activation helps maintain Lipschitz continuity of the vocoder network, as its derivative is bounded by a constant of 1~\cite{ng2025multibandfrequencyreconstructionneural}. Consequently, using Snake can help control the Lipschitz constant $L$ of the vocoder network, limiting the impact of errors from the first-stage mel-spectrogram. Also, in the discriminator we choose to use spectral normalization ~\cite{miyato2018spectralnormalizationgenerativeadversarial} instead of batch normalization, which is designed to guarantee  Lipschitz continuity in discriminator.

\subsubsection{Training Objectives}
Following the Vocos framework, our second-stage training objective consists of two key components: (i) fine-tuning the decoder from the first stage to better align the latent codes with acoustic features, and (ii) training a vocoder that translates mel-spectral codes into time-domain waveforms. To achieve this objective, we employ a combination of loss functions.

 \textbf{Reconstruction Loss.}
Reconstruction loss refers to the L1 distance between the mel-scaled magnitude spectrograms of the ground-truth waveform and the generated waveform ~\cite{kong2020hifigangenerativeadversarialnetworks}. Unlike ~\cite{yang2023hificodecgroupresidualvectorquantization} that uses
80 mel-spectrogram bins, our setup constrains the number of bins to be no smaller than the mel-spectrogram resolution defined in the first stage, which is 96 mel-spectrogram bins for music and environment sound. This ensures consistency between the first stage and second stage, preventing the loss of high-frequency details when training second stage.

 \textbf{Feature Matching Loss.}
Feature matching loss measures the learned similarity between a real and generated sample via discriminator features~\cite{larsen2016autoencodingpixelsusinglearned}, ~\cite{kumar2019melgangenerativeadversarialnetworks} Hifi-GAN ~\cite{kong2020hifigangenerativeadversarialnetworks} first used it as an additional loss to train the generator of vocoder. In our case, feature matching is used to reduce over-smoothing, serving a similar role as the VGG loss applied in the first stage. Unlike the original setting in Hifi-GAN, where the feature matching loss weight is 2, we increase it to 5.

\textbf{Adversarial Loss.}
We employ two discriminators—a multi-resolution discriminator (MRD) ~\cite{kumar2023highfidelityaudiocompressionimproved} and a multi-period discriminator (MPD)—to enhance perceptual quality via adversarial learning~\cite{zeghidour2021soundstreamendtoendneuralaudio}.
\section{Experiments}
\label{others}
Reconstruct waveform from discrete tokens has become a fundamental task for audio codecs. In this section, we assess the performance of method
 relative to established baseline codecs.

 \textbf{DataSets.} The model is trained on the AudioSet dataset, using the entire balanced training subset (bal train), musdb18-hq and the HQ-Conversations dataset~\cite{magicdata2024iscslp}. The AudioSet covers a wide range of sounds, including human and animal vocalizations, musical instruments and genres, as well as common everyday environmental noises. We evaluate the reconstruction performance of the audio domain on the AudioSet evaluation set. 
 We evaluate the reconstruction performance of the music domain on the mixture sound in musdb18hq test set. All audio files are kept at their original sampling rate of 44 kHz. For each audio sample, mel-scaled spectrograms are computed with the following parameters: FFT size $n_{fft}$=1024, hop size $hop_n$ =256, and 96 Mel bins.

 \textbf{Training Details.} In the first stage, we train the mel-spectrogram tokenizer using a combination of L1 reconstruction loss, quantization loss, and perceptual loss until convergence. Afterward, we replace the perceptual loss with a Gram-matrix loss to fine-tune the model, continuing training until convergence. During training,  samples are randomly cropped to 24,320 samples, yielding a mel-spectrogram resolution of 96 × 96. We also train a vocoder using ground-truth mel-spectrograms as reference to evaluate the effect of different loss terms. In the second stage, the encoder and quantizer parameters are frozen. We jointly train the tokenizer decoder, vocoder, and discriminator.

 \textbf{Baseline Methods.} Our proposed model is compared against DAC ~\cite{kumar2023highfidelityaudiocompressionimproved}, SNAC~\cite{siuzdak2024snacmultiscaleneuralaudio}, Spectral Codec ~\cite{langman2025spectralcodecsimprovingnonautoregressive},  NVIDIA NeMo Audio Codec, UniCodec ~\cite{jiang2025unicodecunifiedaudiocodec} and WavTokenizer ~\cite{ji2025wavtokenizerefficientacousticdiscrete}. For all baselines, we use the officially released pretrained checkpoints— the 24kHz version for WavTokenizer and 44 kHz versions for other methods, which are publicly available online. 

\subsection{Evaluation}
We evaluate our models using five primary objective metrics, VISQOL, LSD, Mel Distance, STFT Distance, and Mel Cepstral Distortion. We also performed subjective evaluation with \textbf{MUSHRA}. The primary metrics assess spectral and perceptual fidelity.

 \textbf{VISQOL.} ViSQOL is an objective perceptual audio quality metric that compares reference and degraded audio signals to produce scores correlated with human listening judgments. In this work, we use audio mode, which operates on fullband audio at 48 kHz.

 \textbf{LSD.}
Log-Spectral Distance (LSD) is a widely used objective metric that measures the difference between the log-magnitude spectra of reference and synthesized audio, providing an indication of spectral distortion and overall reconstruction fidelity.

 \textbf{Mel Distance.}
L1 distance between the mel-scaled magnitude spectrograms of the ground truth and the generated sample.

 \textbf{STFT Distance.} L1 distance between time-frequency representations of the ground truth and the prediction, computed using multiscale Short-Time Fourier Transform (STFT).

 \textbf{MCD.} Mel Cepstral Distortion (MCD) is a measure of how different two sequences of mel cepstra are. The core idea is that the smaller the MCD between the original and reconstructed mel cepstral sequences, the more accurately the codec preserves the spectral characteristics and perceptual quality of the original audio signal.

For audio dataset, we use \textbf{Frechet Audio Distance (FAD)}, which is a distribution-level metric that measures how similar reconstruted audio is to a input corpus by comparing their distributions in a deep embedding space. In our setup, we extract embeddings using a pretrained CLAP model \cite{wu2024largescalecontrastivelanguageaudiopretraining} with a sample rate of 48 kHz.

For music dataset, we also use another two metrics:
    We calculate the \textbf{similarity} between the reconstruted vocal part and the original vocal part with a pre-trained wavlm-base-sv ~\cite{Chen_2022} and use Whisper-tiny ~\cite{radford2022robustspeechrecognitionlargescale} to compute the \textbf{Word Error Rate (WER)} of the generated vocal part as the vocal clarity evaluation metric.

\subsection{Ablation Experiment Result for Mel-spectrogram Reconstruction}

To investigate the impact of 2D tokenizer used in first stage on the perceptual quality of the generated audio, we conduct an ablation study using a fixed pretrained vocoder. Specifically, we input the reconstructed mel-spectrograms obtained from 1D and 2D tokenizer into a vocoder trained with ground truth mel-spectrograms.




\begin{table}[h]
    \centering
    \caption{Ablation Experiment One: Comparison of 1d tokenizer and 2d tokenizer used in the first stage. }
    \label{tab:placeholder_label}
    \begin{tabular}{lccccc}
        \toprule
        & \textbf{Visqol $\uparrow$} & \textbf{LSD$\downarrow$} & \textbf{MCD$\downarrow$} & \textbf{STFT Dis$\downarrow$} & \textbf{Mel Dis$\downarrow$} \\
        \midrule
        1D Tokenizer & 3.56 & 1.09 & 11.47 & 3.97 & 1.28 \\
        2D Tokenizer & \textbf{4.36} & \textbf{0.67} & \textbf{2.64} & \textbf{1.68} & \textbf{0.48} \\
        \bottomrule
    \end{tabular}
    \label{tab:2d_comparison}
\end{table}

Table~\ref{tab:2d_comparison} compares different tokenizers used in the first stage. This demonstrates that incorporating using 2d tokenizer can improves the mel-spctrogram reconstruction, which is beneficial for audio reconstruction in the second stage. We provide the encoder architecture details of 1D tokenizer and 2D tokenizer in appendix ~\ref{tabel: ablation_study_conv}.

\begin{table}[h]
    \centering
    \caption{Ablation Experiment Two: Comparison of different loss terms used in the first stage. }
    \label{tab:placeholder_label}
    \begin{tabular}{lccccc}
        \toprule
        & \textbf{ Visqol $\uparrow$} & \textbf{LSD$\downarrow$} & \textbf{MCD$\downarrow$} & \textbf{STFT Dis$\downarrow$} & \textbf{Mel Dis$\downarrow$} \\
        \midrule
        w/o vgg & 4.18 & 0.76 & 3.82 & \textbf{1.87} & 0.57 \\
        w/ vgg   & \textbf{4.29} & \textbf{0.66} & \textbf{3.74} & 1.90 & \textbf{0.56} \\
        \bottomrule
    \end{tabular}
    \label{tab:loss_comparison}
\end{table}

Table~\ref{tab:loss_comparison} compares different loss terms used in the first stage. This demonstrates that incorporating VGG loss effectively mitigates over-smoothing and enhances spectral reconstruction , which is beneficial for training in second stage.

\subsection{Audio Reconstruction}

\setlength{\tabcolsep}{4pt}
\begin{table}[h]
\caption{Objective evaluation metrics for MelTok and baselines on Audioset eval set. We \textbf{bold} the best results in all the models, and \underline{\textbf{bold and underline}} the best results in single-codebook codec models. (24k: ground truth resampled to 24kHz before measurement)}
\centering
\small
\setlength{\tabcolsep}{3pt}
\begin{tabular}{lccccccccl}
\toprule
\textbf{Codec} & \textbf{\makecell{Codebook \\ Number}} & \textbf{\makecell{Bitrate\\(kbps)}} & \textbf{TPS} & \textbf{LSD}$\downarrow$ & \textbf{\makecell{STFT \\ Dis}}${\downarrow}$ & \textbf{\makecell{Mel \\ Dis}}${\downarrow}$ & \textbf{\makecell{Visqol}}${\uparrow}$ & \textbf{MCD}${\downarrow}$  &\textbf{FAD}${\downarrow}$ \\
\midrule
DAC & 9 & 9.0 & 774 & 0.81 & \textbf{2.06} & 0.50 & \textbf{4.19} & 3.90  &0.03 \\
SNAC & 4 & 2.6 & 240 & 0.86 & 2.07 & 0.61 & 4.01 & 3.80  &0.03 \\
Spectral Codec & 9 & 6.9 & 672 & 0.99 & 2.42 & 0.72 & 4.02 & 4.82  &0.02 \\
Nvidia Codec & 9 & 6.9 & 672 & 1.07 & 2.45 & 0.74 & 3.96 & 5.39  &0.08 \\
Wavtokenizer (audio/music) & 1 & 0.9 & 75 & 1.18 & 2.19 & 0.45(24k) & 3.26 & 4.88  &0.15 \\
Wavtokenizer (unify) & 1 & 0.48 & 40 & 1.24 & 2.37 & 0.52(24k) & 3.26 & 5.32  &0.16 \\
UniCodec & 1 & 0.9 & 75 & 1.24 & 2.11 & \textbf{0.38(24k)} & 3.18 & 4.97  &0.12 \\
Our Method & 1 & 3.4 & 260 & \textbf{0.77} & \underline{\textbf{2.07}} & \textbf{0.44} & \underline{\textbf{4.17}} & \textbf{2.92}  &\underline{\textbf{0.12}} \\
\bottomrule
\end{tabular}
\label{tab:codec_metrics}
\end{table}

In the second stage, we use the mel-spectrogram tokens obtained from the first stage as input. After training the second stage, our final results are obtained from the jointly optimized decoder and vocoder.  Our evaluation, conducted on the AudioSet test set and detailed in ~\ref{tab:codec_metrics}.  The result shows that Mel Tok achieves competitive perceptual quality (VISQOL 4.17) while using only a single codebook. It also obtains the best fidelity metrics (lowest LSD and MCD Distance), outperforming single-codebook baseline and multi-codebook baselines.

\paragraph{Subject Evaluation.}
We conducted a MUSHRA-style listening test to evaluate the perceptual quality of the generated audio. A total of 15 participants were recruited for the experiment. Each participant was presented with 30 randomly selected audio samples drawn from a diverse set of test cases including music, speech, and general sounds. For each trial, participants were asked to rate the audio samples on a continuous quality scale, following the MUSHRA protocol. After the test, we aggregated the ratings across all participants and samples to obtain the final statistical results. The subjective evaluation indicates that Mel Tok achieves perceptual quality comparable to other multi-codec approaches.

\begin{figure}
    \centering
    \includegraphics[width=0.5\linewidth]{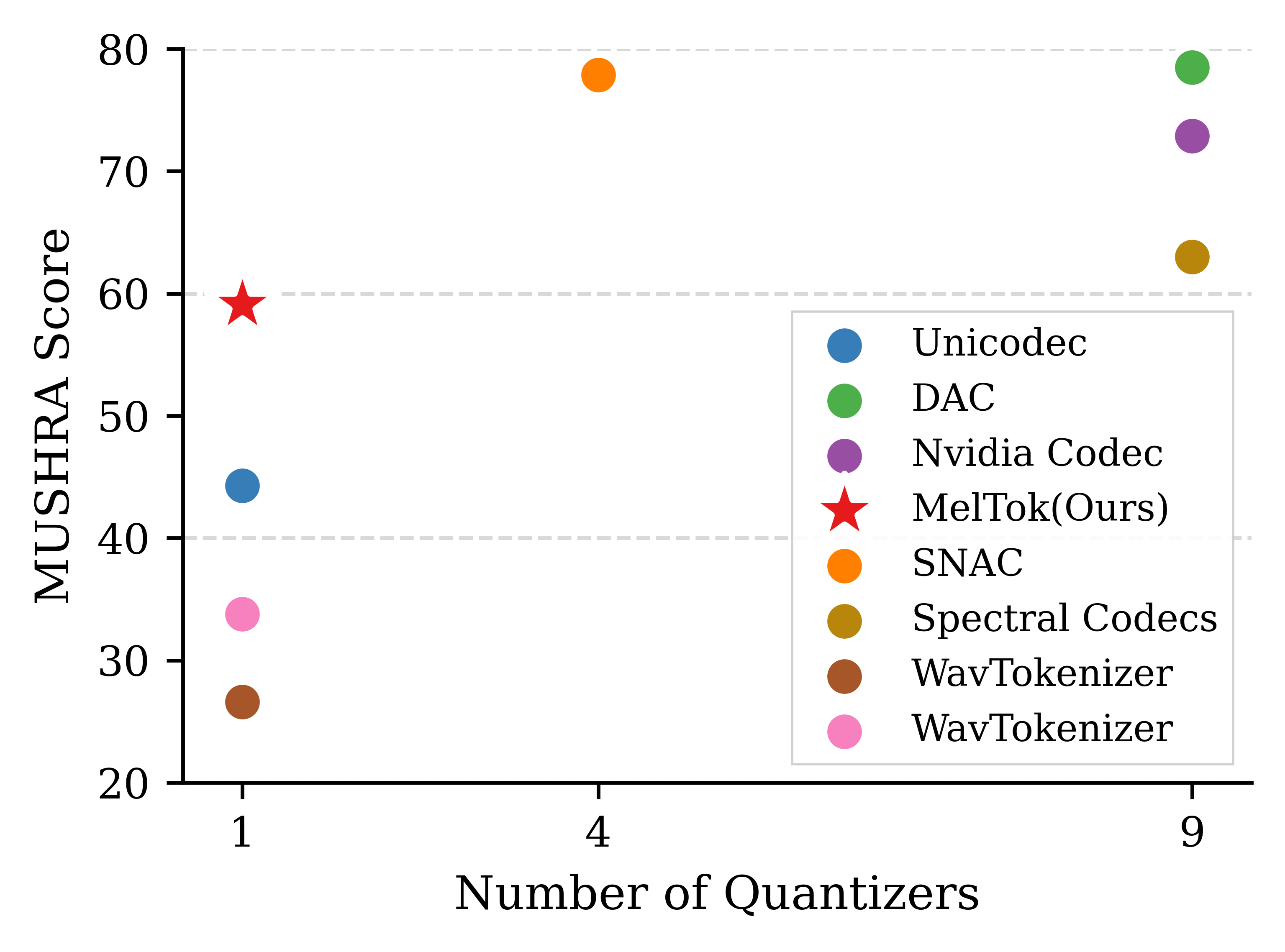}
\caption{Subjective evaluation metrics calculated for different codecs. Points closer to the top-left indicate that better perceptual quality is achieved using fewer tokens, corresponding to better codec performance.} 
    \label{fig:subject_result}
\end{figure}


\setlength{\tabcolsep}{4pt}
\begin{table}[h]
    \centering
    \caption{Objective quality on the \emph{musdb18-hq test} dataset.We \textbf{bold} the best results in all the models, and \underline{\textbf{bold and underline}} the best results in
single-codebook codec models. “Original” refers to results obtained
using the ground-truth waveform. (24k: ground truth resampled to 24kHz before measurement)}
    \label{tab:placeholder_label}
    \begin{tabular}{lccccccc}
        \midrule
        \textbf{Codec} & \textbf{LSD$\downarrow$} & \textbf{STFT Dis$\downarrow$} & \textbf{Mel Dis$\downarrow$} & 
        \textbf{Visqol$\uparrow$} & \textbf{MCD$\downarrow$} & 
        \textbf{\makecell{WER$\downarrow$\\Original: 1.62}}  & \textbf{SPK-SIM$\uparrow$} \\
        \midrule
        DAC                             & \textbf{0.85} & \underline{\textbf{1.71}} & \textbf{0.36} & 4.18 & \textbf{2.04} & \textbf{1.85} & \textbf{0.98} \\
        SNAC                            & 1.20 & 2.53 & 0.53 & 3.80 & 3.08 & 2.10 & 0.96 \\
        Spectral Codec                  & 1.65 & 3.17 & 0.73 & 3.92 & 3.12 & 2.46 & 0.90 \\
        Nvidia Codec                    & 1.77 & 3.24 & 0.76 & 3.84 & 3.28 & 2.22 & 0.92 \\
        \makecell[l]{Wavtokenizer \\ (audio/music)}  & 2.01 & 2.20 & 0.47 (24k) & 2.73 & 3.90 & 3.31 & 0.92 \\
        \makecell[l]{Wavtokenizer \\ (unify)}
                                        & 2.04 & 2.39 & 0.62 (24k) &2.69 & 4.38 &3.06 & 0.89 \\
        UniCodec                        & 1.99 & 2.12 & \textbf{0.40 (24k)} & 2.83 & 3.26 & 2.34 & 0.93 \\
        Our Method                      & \underline{\textbf{0.90}}& \textbf{2.02} & \textbf{0.40} & \textbf{4.23} & \underline{\textbf{2.46}} & \underline{\textbf{2.14}} & \underline{\textbf{0.96}} \\
        \hline
    \end{tabular}
    \label{tab:codec_metrics_music}
\end{table}

\begin{table}[h]
    \caption{Downstream sound event classification performance. "Reference" refers to results obtained using the ground-truth waveform, while the other columns show the performance using audio reconstructed by different models.}
    \centering
    \begin{tabular}{lcccc}
        \hline
        & Reference & SNAC & DAC & Our Method \\
        \hline
        \textbf{F1}$\uparrow$  & 0.3899 & 0.3223 & 0.3363 & \textbf{0.3398} \\
        \textbf{mAP}$\uparrow$ & 0.1626 & 0.1278 & 0.1251 & \textbf{0.1345} \\
        \hline
    \end{tabular}
    \label{tab:downstream}
\end{table}

\begin{table}[h]
    \centering
    \caption{Downstream music genre classification performance.}
    \label{tab:placeholder_label}
    \begin{tabular}{lccc}
        \toprule
        & \textbf{Accuracy$\uparrow$} & \textbf{Weighted F1 Score$\uparrow$} & \textbf{Weighted Precision$\uparrow$} \\
        \midrule
        Reference              & 0.63 & 0.63 & 0.58 \\
        DAC                   & 0.60 & 0.57 & \textbf{0.57} \\
        SNAC                  & \textbf{0.63} & \textbf{0.59} & 0.57 \\
        Nvidia Codec          & 0.60 & 0.57 & \textbf{0.57} \\
        Spectral Codec             & 0.65 & 0.59 & \textbf{0.57} \\
        WavTokenizer (audio)  & 0.27 & 0.15 & 0.15 \\
        WavTokenizer (unify)  & 0.24 & 0.12 & 0.18 \\
        UniCodec              & 0.26 & 0.15 & 0.25 \\
        Our Method            & \underline{\textbf{0.57}} & \underline{\textbf{0.50}} & \underline{\textbf{0.50}} \\
        \bottomrule
    \end{tabular}
    \label{tab:genre}
\end{table}

A key property of a codec is its ability to compress and reconstruct unseen data. After augmenting the training set with hq-conversations ~\cite{magicdata2024iscslp}, we tested the model on unseen music data. The results, as reported in Table~\ref{tab:codec_metrics_music}, demonstrate that our codec generalizes well beyond AudioSet and maintains competitive perceptual and spectral quality in out-of-distribution data. 

\subsection{Downstream Audio Understanding Task Evaluation}
Unlike speech-only datasets, which can be evaluated using reconstructed waveform quality by ASR models, AudioSet contains multiple sound categories and is designed for audio event classification. We further evaluate our codec on this task. Specifically, we employ pretrained models from~\cite{dinkel2023cedconsistentensembledistillation} and compute the top-3 predicted labels using the reconstructed waveforms. To assess the codec’s ability to preserve semantic information, we report F1 and mAP scores in Table ~\ref{tab:downstream}, which measure the accuracy of sound event classification. The performance demonstrates the codec’s effectiveness in retaining discriminative detail beyond perceptual quality, achieving better downstream classification results compared to other codec baselines.

Also, we performed music genre classification using a fine-tuned version of the pre-trained model DistilHuBERT ~\cite{chang2022distilhubertspeechrepresentationlearning} on the GTZAN dataset ~\cite{tzanetakis_essl_cook_2001}. As shown in Table~\ref{tab:genre}, our method achieves the best performance among single-codebook codecs, while remaining competitive with multi-codebook codecs such as Spectral Codec. This demonstrates that our approach provides an effective representation for downstream music understanding tasks.

\newpage
\section{Limitations and Future Work}
Our experiments reveal that training MelTok is disrupted when using low-resolution data that has been artificially upsampled to a high sampling rate. Currently, open-source high-fidelity audio datasets are scarce, and even 44kHz/48kHz corpora often include upsampled 16kHz content, limiting compression performance. In future work, we plan to collect more high-resolution recordings, enhance our model for better perceptual quality, and develop suitable benchmarks for evaluating non-speech, non-music sounds.

Although our experiments demonstrate that preserving more acoustic detail in the codec is promising for audio understanding tasks, its impact on downstream generation tasks remains unknown. Due to space and computational constraints, we have focused on showcasing MelTok’s reconstruction capabilities and have not yet use tokens obtained from MelTok to train in Audio Language Models (ALM). In future work, we plan to investigate the performance of MelTok-based ALMs on downstream audio tasks.

\section{Conclusion}
In this paper, we introduce MelTok, a neural codec that compresses full-bandwidth audio into a single codebook. We achieve this by converting mel-spectrograms into discrete tokens using a 2D tokenizer, and then reconstructing high-quality audio from these tokens with a vocoder. In the process, we investigate the disentangled distributions obtained from 2D tokenization. These findings offer new insights for future research in audio codec development. Specifically: First, incorporating perceptual loss into the reconstruction stage helps preserve more structure detail in the codebook. Second, two-stage training stabilizes the overall training process, leading to faster convergence than GAN-based models. Third, using mel-spectrograms as the representation better preserves frequency details during compression.

\newpage

\bibliography{iclr2026_conference}

\begin{thebibliography}{61}
\providecommand{\natexlab}[1]{#1}
\providecommand{\url}[1]{\texttt{#1}}
\expandafter\ifx\csname urlstyle\endcsname\relax
  \providecommand{\doi}[1]{doi: #1}\else
  \providecommand{\doi}{doi: \begingroup \urlstyle{rm}\Url}\fi

\bibitem[200(2001)]{2001SpokenLP}
Spoken language processing: A guide to theory, algorithm and system development.
\newblock 2001.

\bibitem[Agustsson et~al.(2017)Agustsson, Mentzer, Tschannen, Cavigelli, Timofte, Benini, and Gool]{agustsson2017softtohardvectorquantizationendtoend}
Eirikur Agustsson, Fabian Mentzer, Michael Tschannen, Lukas Cavigelli, Radu Timofte, Luca Benini, and Luc~Van Gool.
\newblock Soft-to-hard vector quantization for end-to-end learning compressible representations, 2017.

\bibitem[Ai et~al.(2024)Ai, Jiang, Lu, Du, and Ling]{ai2024apcodecneuralaudiocodec}
Yang Ai, Xiao-Hang Jiang, Ye-Xin Lu, Hui-Peng Du, and Zhen-Hua Ling.
\newblock Apcodec: A neural audio codec with parallel amplitude and phase spectrum encoding and decoding, 2024.

\bibitem[Bahuleyan(2018)]{bahuleyan2018musicgenreclassificationusing}
Hareesh Bahuleyan.
\newblock Music genre classification using machine learning techniques, 2018.
\newblock URL \url{https://arxiv.org/abs/1804.01149}.

\bibitem[Cakir et~al.(2015)Cakir, Heittola, Huttunen, and Virtanen]{7280624}
Emre Cakir, Toni Heittola, Heikki Huttunen, and Tuomas Virtanen.
\newblock Polyphonic sound event detection using multi label deep neural networks.
\newblock In \emph{2015 International Joint Conference on Neural Networks (IJCNN)}, pages 1--7, 2015.
\newblock \doi{10.1109/IJCNN.2015.7280624}.

\bibitem[Chang et~al.(2022)Chang, wen Yang, and yi~Lee]{chang2022distilhubertspeechrepresentationlearning}
Heng-Jui Chang, Shu wen Yang, and Hung yi~Lee.
\newblock Distilhubert: Speech representation learning by layer-wise distillation of hidden-unit bert, 2022.
\newblock URL \url{https://arxiv.org/abs/2110.01900}.

\bibitem[Chen et~al.(2022)Chen, Wang, Chen, Wu, Liu, Chen, Li, Kanda, Yoshioka, Xiao, Wu, Zhou, Ren, Qian, Qian, Wu, Zeng, Yu, and Wei]{Chen_2022}
Sanyuan Chen, Chengyi Wang, Zhengyang Chen, Yu~Wu, Shujie Liu, Zhuo Chen, Jinyu Li, Naoyuki Kanda, Takuya Yoshioka, Xiong Xiao, Jian Wu, Long Zhou, Shuo Ren, Yanmin Qian, Yao Qian, Jian Wu, Michael Zeng, Xiangzhan Yu, and Furu Wei.
\newblock Wavlm: Large-scale self-supervised pre-training for full stack speech processing.
\newblock \emph{IEEE Journal of Selected Topics in Signal Processing}, 16\penalty0 (6):\penalty0 1505–1518, October 2022.
\newblock ISSN 1941-0484.
\newblock \doi{10.1109/jstsp.2022.3188113}.
\newblock URL \url{http://dx.doi.org/10.1109/JSTSP.2022.3188113}.

\bibitem[Dinkel et~al.(2023)Dinkel, Wang, Yan, Zhang, and Wang]{dinkel2023cedconsistentensembledistillation}
Heinrich Dinkel, Yongqing Wang, Zhiyong Yan, Junbo Zhang, and Yujun Wang.
\newblock Ced: Consistent ensemble distillation for audio tagging, 2023.

\bibitem[Dosovitskiy et~al.(2021)Dosovitskiy, Beyer, Kolesnikov, Weissenborn, Zhai, Unterthiner, Dehghani, Minderer, Heigold, Gelly, Uszkoreit, and Houlsby]{dosovitskiy2021imageworth16x16words}
Alexey Dosovitskiy, Lucas Beyer, Alexander Kolesnikov, Dirk Weissenborn, Xiaohua Zhai, Thomas Unterthiner, Mostafa Dehghani, Matthias Minderer, Georg Heigold, Sylvain Gelly, Jakob Uszkoreit, and Neil Houlsby.
\newblock An image is worth 16x16 words: Transformers for image recognition at scale, 2021.

\bibitem[Du et~al.(2024{\natexlab{a}})Du, Ai, Zheng, and Ling]{du2024apcodecspectrumcodingbasedhighfidelityhighcompressionrate}
Hui-Peng Du, Yang Ai, Rui-Chen Zheng, and Zhen-Hua Ling.
\newblock Apcodec+: A spectrum-coding-based high-fidelity and high-compression-rate neural audio codec with staged training paradigm, 2024{\natexlab{a}}.

\bibitem[Du et~al.(2023)Du, Zhang, Hu, and Zheng]{du2023funcodecfundamentalreproducibleintegrable}
Zhihao Du, Shiliang Zhang, Kai Hu, and Siqi Zheng.
\newblock Funcodec: A fundamental, reproducible and integrable open-source toolkit for neural speech codec, 2023.
\newblock URL \url{https://arxiv.org/abs/2309.07405}.

\bibitem[Du et~al.(2024{\natexlab{b}})Du, Chen, Zhang, Hu, Lu, Yang, Hu, Zheng, Gu, Ma, Gao, and Yan]{du2024cosyvoicescalablemultilingualzeroshot}
Zhihao Du, Qian Chen, Shiliang Zhang, Kai Hu, Heng Lu, Yexin Yang, Hangrui Hu, Siqi Zheng, Yue Gu, Ziyang Ma, Zhifu Gao, and Zhijie Yan.
\newblock Cosyvoice: A scalable multilingual zero-shot text-to-speech synthesizer based on supervised semantic tokens, 2024{\natexlab{b}}.

\bibitem[Défossez et~al.(2022)Défossez, Copet, Synnaeve, and Adi]{défossez2022highfidelityneuralaudio}
Alexandre Défossez, Jade Copet, Gabriel Synnaeve, and Yossi Adi.
\newblock High fidelity neural audio compression, 2022.

\bibitem[Esser et~al.(2021)Esser, Rombach, and Ommer]{esser2021tamingtransformershighresolutionimage}
Patrick Esser, Robin Rombach, and Björn Ommer.
\newblock Taming transformers for high-resolution image synthesis, 2021.

\bibitem[Gatys et~al.(2016)Gatys, Ecker, and Bethge]{7780634}
Leon~A. Gatys, Alexander~S. Ecker, and Matthias Bethge.
\newblock Image style transfer using convolutional neural networks.
\newblock In \emph{2016 IEEE Conference on Computer Vision and Pattern Recognition (CVPR)}, pages 2414--2423, 2016.
\newblock \doi{10.1109/CVPR.2016.265}.

\bibitem[Guo et~al.(2025)Guo, Li, Wang, Li, Shao, Zhang, Du, Chen, Liu, and Yu]{guo2025recentadvancesdiscretespeech}
Yiwei Guo, Zhihan Li, Hankun Wang, Bohan Li, Chongtian Shao, Hanglei Zhang, Chenpeng Du, Xie Chen, Shujie Liu, and Kai Yu.
\newblock Recent advances in discrete speech tokens: A review, 2025.
\newblock URL \url{https://arxiv.org/abs/2502.06490}.

\bibitem[Hsu et~al.(2021)Hsu, Bolte, Tsai, Lakhotia, Salakhutdinov, and Mohamed]{hsu2021hubertselfsupervisedspeechrepresentation}
Wei-Ning Hsu, Benjamin Bolte, Yao-Hung~Hubert Tsai, Kushal Lakhotia, Ruslan Salakhutdinov, and Abdelrahman Mohamed.
\newblock Hubert: Self-supervised speech representation learning by masked prediction of hidden units, 2021.

\bibitem[Ji et~al.(2025)Ji, Jiang, Wang, Chen, Fang, Zuo, Yang, Cheng, Wang, Li, Zhang, Yang, Huang, Jiang, Chen, Zheng, and Zhao]{ji2025wavtokenizerefficientacousticdiscrete}
Shengpeng Ji, Ziyue Jiang, Wen Wang, Yifu Chen, Minghui Fang, Jialong Zuo, Qian Yang, Xize Cheng, Zehan Wang, Ruiqi Li, Ziang Zhang, Xiaoda Yang, Rongjie Huang, Yidi Jiang, Qian Chen, Siqi Zheng, and Zhou Zhao.
\newblock Wavtokenizer: an efficient acoustic discrete codec tokenizer for audio language modeling, 2025.

\bibitem[Jiang et~al.(2025)Jiang, Chen, Ji, Xi, Wang, Zhang, Yue, Zhang, and Li]{jiang2025unicodecunifiedaudiocodec}
Yidi Jiang, Qian Chen, Shengpeng Ji, Yu~Xi, Wen Wang, Chong Zhang, Xianghu Yue, ShiLiang Zhang, and Haizhou Li.
\newblock Unicodec: Unified audio codec with single domain-adaptive codebook, 2025.
\newblock URL \url{https://arxiv.org/abs/2502.20067}.

\bibitem[Jiao et~al.(2021)Jiao, Gabrys, Tinchev, Putrycz, Korzekwa, and Klimkov]{jiao2021universalneuralvocodingparallel}
Yunlong Jiao, Adam Gabrys, Georgi Tinchev, Bartosz Putrycz, Daniel Korzekwa, and Viacheslav Klimkov.
\newblock Universal neural vocoding with parallel wavenet, 2021.

\bibitem[Juang and Gray(1982)]{Juang1982MultipleSV}
Biing-Hwang Juang and Augustine~H. Gray.
\newblock Multiple stage vector quantization for speech coding.
\newblock In \emph{IEEE International Conference on Acoustics, Speech, and Signal Processing}, 1982.

\bibitem[Kong et~al.(2020{\natexlab{a}})Kong, Kim, and Bae]{kong2020hifigangenerativeadversarialnetworks}
Jungil Kong, Jaehyeon Kim, and Jaekyoung Bae.
\newblock Hifi-gan: Generative adversarial networks for efficient and high fidelity speech synthesis, 2020{\natexlab{a}}.

\bibitem[Kong et~al.(2020{\natexlab{b}})Kong, Cao, Iqbal, Wang, Wang, and Plumbley]{kong2020pannslargescalepretrainedaudio}
Qiuqiang Kong, Yin Cao, Turab Iqbal, Yuxuan Wang, Wenwu Wang, and Mark~D. Plumbley.
\newblock Panns: Large-scale pretrained audio neural networks for audio pattern recognition, 2020{\natexlab{b}}.
\newblock URL \url{https://arxiv.org/abs/1912.10211}.

\bibitem[Kumar et~al.(2019)Kumar, Kumar, de~Boissiere, Gestin, Teoh, Sotelo, de~Brebisson, Bengio, and Courville]{kumar2019melgangenerativeadversarialnetworks}
Kundan Kumar, Rithesh Kumar, Thibault de~Boissiere, Lucas Gestin, Wei~Zhen Teoh, Jose Sotelo, Alexandre de~Brebisson, Yoshua Bengio, and Aaron Courville.
\newblock Melgan: Generative adversarial networks for conditional waveform synthesis, 2019.

\bibitem[Kumar et~al.(2023)Kumar, Seetharaman, Luebs, Kumar, and Kumar]{kumar2023highfidelityaudiocompressionimproved}
Rithesh Kumar, Prem Seetharaman, Alejandro Luebs, Ishaan Kumar, and Kundan Kumar.
\newblock High-fidelity audio compression with improved rvqgan, 2023.

\bibitem[Langman et~al.(2025)Langman, Jukić, Dhawan, Koluguri, and Li]{langman2025spectralcodecsimprovingnonautoregressive}
Ryan Langman, Ante Jukić, Kunal Dhawan, Nithin~Rao Koluguri, and Jason Li.
\newblock Spectral codecs: Improving non-autoregressive speech synthesis with spectrogram-based audio codecs, 2025.

\bibitem[Larsen et~al.(2016)Larsen, Sønderby, Larochelle, and Winther]{larsen2016autoencodingpixelsusinglearned}
Anders Boesen~Lindbo Larsen, Søren~Kaae Sønderby, Hugo Larochelle, and Ole Winther.
\newblock Autoencoding beyond pixels using a learned similarity metric, 2016.

\bibitem[Lewtun(2022)]{music_genres}
Lewtun.
\newblock music\_genres, 2022.

\bibitem[Li et~al.(2024)Li, Xue, Guo, Zhu, Lv, Xie, Chen, Yin, and Li]{li2024singlecodecsinglecodebookspeechcodec}
Hanzhao Li, Liumeng Xue, Haohan Guo, Xinfa Zhu, Yuanjun Lv, Lei Xie, Yunlin Chen, Hao Yin, and Zhifei Li.
\newblock Single-codec: Single-codebook speech codec towards high-performance speech generation, 2024.

\bibitem[Libera et~al.(2025)Libera, Paissan, Subakan, and Ravanelli]{dellalibera2025focalcodeclowbitratespeechcoding}
Luca~Della Libera, Francesco Paissan, Cem Subakan, and Mirco Ravanelli.
\newblock Focalcodec: Low-bitrate speech coding via focal modulation networks, 2025.

\bibitem[{Magic Data}(2024)]{magicdata2024iscslp}
{Magic Data}.
\newblock Conversational voice clone challenge (covoc) -- iscslp 2024, 2024.

\bibitem[Miyato et~al.(2018)Miyato, Kataoka, Koyama, and Yoshida]{miyato2018spectralnormalizationgenerativeadversarial}
Takeru Miyato, Toshiki Kataoka, Masanori Koyama, and Yuichi Yoshida.
\newblock Spectral normalization for generative adversarial networks, 2018.

\bibitem[Mousavi et~al.(2025)Mousavi, Maimon, Moumen, Petermann, Shi, Wu, Yang, Kuznetsova, Ploujnikov, Marxer, Ramabhadran, Elizalde, Lugosch, Li, Subakan, Woodland, Kim, yi~Lee, Watanabe, Adi, and Ravanelli]{mousavi2025discreteaudiotokenssurvey}
Pooneh Mousavi, Gallil Maimon, Adel Moumen, Darius Petermann, Jiatong Shi, Haibin Wu, Haici Yang, Anastasia Kuznetsova, Artem Ploujnikov, Ricard Marxer, Bhuvana Ramabhadran, Benjamin Elizalde, Loren Lugosch, Jinyu Li, Cem Subakan, Phil Woodland, Minje Kim, Hung yi~Lee, Shinji Watanabe, Yossi Adi, and Mirco Ravanelli.
\newblock Discrete audio tokens: More than a survey!, 2025.

\bibitem[Ng et~al.(2025)Ng, Zhou, Chao, Xiong, Ma, and Chng]{ng2025multibandfrequencyreconstructionneural}
Dianwen Ng, Kun Zhou, Yi-Wen Chao, Zhiwei Xiong, Bin Ma, and Eng~Siong Chng.
\newblock Multi-band frequency reconstruction for neural psychoacoustic coding, 2025.

\bibitem[NVIDIA et~al.(2025)NVIDIA, :, Agarwal, Ali, Bala, Balaji, Barker, Cai, Chattopadhyay, Chen, Cui, Ding, Dworakowski, Fan, Fenzi, Ferroni, Fidler, Fox, Ge, Ge, Gu, Gururani, He, Huang, Huffman, Jannaty, Jin, Kim, Klár, Lam, Lan, Leal-Taixe, Li, Li, Lin, Lin, Ling, Liu, Liu, Luo, Ma, Mao, Mo, Mousavian, Nah, Niverty, Page, Paschalidou, Patel, Pavao, Ramezanali, Reda, Ren, Sabavat, Schmerling, Shi, Stefaniak, Tang, Tchapmi, Tredak, Tseng, Varghese, Wang, Wang, Wang, Wang, Wei, Wei, Wu, Xu, Yang, Yen-Chen, Zeng, Zeng, Zhang, Zhang, Zhang, Zhao, and Zolkowski]{nvidia2025cosmosworldfoundationmodel}
NVIDIA, :, Niket Agarwal, Arslan Ali, Maciej Bala, Yogesh Balaji, Erik Barker, Tiffany Cai, Prithvijit Chattopadhyay, Yongxin Chen, Yin Cui, Yifan Ding, Daniel Dworakowski, Jiaojiao Fan, Michele Fenzi, Francesco Ferroni, Sanja Fidler, Dieter Fox, Songwei Ge, Yunhao Ge, Jinwei Gu, Siddharth Gururani, Ethan He, Jiahui Huang, Jacob Huffman, Pooya Jannaty, Jingyi Jin, Seung~Wook Kim, Gergely Klár, Grace Lam, Shiyi Lan, Laura Leal-Taixe, Anqi Li, Zhaoshuo Li, Chen-Hsuan Lin, Tsung-Yi Lin, Huan Ling, Ming-Yu Liu, Xian Liu, Alice Luo, Qianli Ma, Hanzi Mao, Kaichun Mo, Arsalan Mousavian, Seungjun Nah, Sriharsha Niverty, David Page, Despoina Paschalidou, Zeeshan Patel, Lindsey Pavao, Morteza Ramezanali, Fitsum Reda, Xiaowei Ren, Vasanth Rao~Naik Sabavat, Ed~Schmerling, Stella Shi, Bartosz Stefaniak, Shitao Tang, Lyne Tchapmi, Przemek Tredak, Wei-Cheng Tseng, Jibin Varghese, Hao Wang, Haoxiang Wang, Heng Wang, Ting-Chun Wang, Fangyin Wei, Xinyue Wei, Jay~Zhangjie Wu, Jiashu Xu, Wei Yang, Lin Yen-Chen, Xiaohui Zeng,
  Yu~Zeng, Jing Zhang, Qinsheng Zhang, Yuxuan Zhang, Qingqing Zhao, and Artur Zolkowski.
\newblock Cosmos world foundation model platform for physical ai, 2025.

\bibitem[Paulavi{\v{c}}ius and {\v{Z}}ilinskas(2006)]{paulavivcius2006analysis}
Remigijus Paulavi{\v{c}}ius and Julius {\v{Z}}ilinskas.
\newblock Analysis of different norms and corresponding lipschitz constants for global optimization.
\newblock \emph{Technological and Economic Development of Economy}, 12\penalty0 (4):\penalty0 301--306, 2006.

\bibitem[Peng et~al.(2024)Peng, Huang, Li, Mohamed, and Harwath]{peng2024voicecraftzeroshotspeechediting}
Puyuan Peng, Po-Yao Huang, Shang-Wen Li, Abdelrahman Mohamed, and David Harwath.
\newblock Voicecraft: Zero-shot speech editing and text-to-speech in the wild, 2024.

\bibitem[Radford et~al.(2022{\natexlab{a}})Radford, Kim, Xu, Brockman, McLeavey, and Sutskever]{radford2022robustspeechrecognitionlargescale}
Alec Radford, Jong~Wook Kim, Tao Xu, Greg Brockman, Christine McLeavey, and Ilya Sutskever.
\newblock Robust speech recognition via large-scale weak supervision, 2022{\natexlab{a}}.

\bibitem[Radford et~al.(2022{\natexlab{b}})Radford, Kim, Xu, Brockman, McLeavey, and Sutskever]{radford2022whisper}
Alec Radford, Jong~Wook Kim, Tao Xu, Greg Brockman, Christine McLeavey, and Ilya Sutskever.
\newblock Robust speech recognition via large-scale weak supervision, 2022{\natexlab{b}}.
\newblock URL \url{https://arxiv.org/abs/2212.04356}.

\bibitem[Sheng and Pavlovskiy(2018)]{sheng2018reducingoversmoothnessspeechsynthesis}
Leyuan Sheng and Evgeniy~N. Pavlovskiy.
\newblock Reducing over-smoothness in speech synthesis using generative adversarial networks, 2018.

\bibitem[Simonyan and Zisserman(2015)]{simonyan2015deepconvolutionalnetworkslargescale}
Karen Simonyan and Andrew Zisserman.
\newblock Very deep convolutional networks for large-scale image recognition, 2015.

\bibitem[Siuzdak(2024)]{siuzdak2024vocosclosinggaptimedomain}
Hubert Siuzdak.
\newblock Vocos: Closing the gap between time-domain and fourier-based neural vocoders for high-quality audio synthesis, 2024.

\bibitem[Siuzdak et~al.(2024)Siuzdak, Grötschla, and Lanzendörfer]{siuzdak2024snacmultiscaleneuralaudio}
Hubert Siuzdak, Florian Grötschla, and Luca~A. Lanzendörfer.
\newblock Snac: Multi-scale neural audio codec, 2024.

\bibitem[Theis et~al.(2017)Theis, Shi, Cunningham, and Huszár]{theis2017lossyimagecompressioncompressive}
Lucas Theis, Wenzhe Shi, Andrew Cunningham, and Ferenc Huszár.
\newblock Lossy image compression with compressive autoencoders, 2017.

\bibitem[Tzanetakis et~al.(2001)Tzanetakis, Essl, and Cook]{tzanetakis_essl_cook_2001}
George Tzanetakis, Georg Essl, and Perry Cook.
\newblock Automatic musical genre classification of audio signals, 2001.
\newblock URL \url{http://ismir2001.ismir.net/pdf/tzanetakis.pdf}.

\bibitem[Ulicny et~al.(2022)Ulicny, Krylov, and Dahyot]{ulicny2022harmonicconvolutionalnetworksbased}
Matej Ulicny, Vladimir~A. Krylov, and Rozenn Dahyot.
\newblock Harmonic convolutional networks based on discrete cosine transform, 2022.
\newblock URL \url{https://arxiv.org/abs/2001.06570}.

\bibitem[Valin et~al.(2016)Valin, Terriberry, and Maxwell]{valin2016fullbandwidthaudiocodeclow}
Jean-Marc Valin, Timothy~B. Terriberry, and Gregory Maxwell.
\newblock A full-bandwidth audio codec with low complexity and very low delay, 2016.
\newblock URL \url{https://arxiv.org/abs/1602.05311}.

\bibitem[van~den Oord et~al.(2016)van~den Oord, Dieleman, Zen, Simonyan, Vinyals, Graves, Kalchbrenner, Senior, and Kavukcuoglu]{oord2016wavenetgenerativemodelraw}
Aaron van~den Oord, Sander Dieleman, Heiga Zen, Karen Simonyan, Oriol Vinyals, Alex Graves, Nal Kalchbrenner, Andrew Senior, and Koray Kavukcuoglu.
\newblock Wavenet: A generative model for raw audio, 2016.

\bibitem[van~den Oord et~al.(2018)van~den Oord, Vinyals, and Kavukcuoglu]{oord2018neuraldiscreterepresentationlearning}
Aaron van~den Oord, Oriol Vinyals, and Koray Kavukcuoglu.
\newblock Neural discrete representation learning, 2018.

\bibitem[Wu et~al.(2023)Wu, Gebru, Marković, and Richard]{Wu_2023}
Yi-Chiao Wu, Israel~D. Gebru, Dejan Marković, and Alexander Richard.
\newblock Audiodec: An open-source streaming high-fidelity neural audio codec.
\newblock In \emph{ICASSP 2023 - 2023 IEEE International Conference on Acoustics, Speech and Signal Processing (ICASSP)}, page 1–5. IEEE, June 2023.
\newblock \doi{10.1109/icassp49357.2023.10096509}.

\bibitem[Wu et~al.(2024{\natexlab{a}})Wu, Marković, Krenn, Gebru, and Richard]{wu2024scoredecphasepreservinghighfidelityaudio}
Yi-Chiao Wu, Dejan Marković, Steven Krenn, Israel~D. Gebru, and Alexander Richard.
\newblock Scoredec: A phase-preserving high-fidelity audio codec with a generalized score-based diffusion post-filter, 2024{\natexlab{a}}.

\bibitem[Wu et~al.(2024{\natexlab{b}})Wu, Chen, Zhang, Hui, Nezhurina, Berg-Kirkpatrick, and Dubnov]{wu2024largescalecontrastivelanguageaudiopretraining}
Yusong Wu, Ke~Chen, Tianyu Zhang, Yuchen Hui, Marianna Nezhurina, Taylor Berg-Kirkpatrick, and Shlomo Dubnov.
\newblock Large-scale contrastive language-audio pretraining with feature fusion and keyword-to-caption augmentation, 2024{\natexlab{b}}.
\newblock URL \url{https://arxiv.org/abs/2211.06687}.

\bibitem[Yang et~al.(2023)Yang, Liu, Huang, Tian, Weng, and Zou]{yang2023hificodecgroupresidualvectorquantization}
Dongchao Yang, Songxiang Liu, Rongjie Huang, Jinchuan Tian, Chao Weng, and Yuexian Zou.
\newblock Hifi-codec: Group-residual vector quantization for high fidelity audio codec, 2023.

\bibitem[Yang et~al.(2024)Yang, Guo, Wang, Huang, Li, Tan, Wu, and Meng]{yang2024uniaudio15largelanguage}
Dongchao Yang, Haohan Guo, Yuanyuan Wang, Rongjie Huang, Xiang Li, Xu~Tan, Xixin Wu, and Helen Meng.
\newblock Uniaudio 1.5: Large language model-driven audio codec is a few-shot audio task learner, 2024.

\bibitem[Yang et~al.(2020)Yang, Zhen, Beack, and Kim]{yang2020sourceawareneuralspeechcoding}
Haici Yang, Kai Zhen, Seungkwon Beack, and Minje Kim.
\newblock Source-aware neural speech coding for noisy speech compression, 2020.

\bibitem[Ye et~al.(2024)Ye, Sun, Lei, Lin, Tan, Dai, Kong, Chen, Pan, Liu, Guo, and Xue]{ye2024codecdoesmatterexploring}
Zhen Ye, Peiwen Sun, Jiahe Lei, Hongzhan Lin, Xu~Tan, Zheqi Dai, Qiuqiang Kong, Jianyi Chen, Jiahao Pan, Qifeng Liu, Yike Guo, and Wei Xue.
\newblock Codec does matter: Exploring the semantic shortcoming of codec for audio language model, 2024.

\bibitem[Ye et~al.(2025)Ye, Zhu, Chan, Wang, Tan, Lei, Peng, Liu, Jin, Dai, Lin, Chen, Du, Xue, Chen, Li, Xie, Kong, Guo, and Xue]{ye2025llasascalingtraintimeinferencetime}
Zhen Ye, Xinfa Zhu, Chi-Min Chan, Xinsheng Wang, Xu~Tan, Jiahe Lei, Yi~Peng, Haohe Liu, Yizhu Jin, Zheqi Dai, Hongzhan Lin, Jianyi Chen, Xingjian Du, Liumeng Xue, Yunlin Chen, Zhifei Li, Lei Xie, Qiuqiang Kong, Yike Guo, and Wei Xue.
\newblock Llasa: Scaling train-time and inference-time compute for llama-based speech synthesis, 2025.

\bibitem[Yu et~al.(2024)Yu, Weber, Deng, Shen, Cremers, and Chen]{yu2024imageworth32tokens}
Qihang Yu, Mark Weber, Xueqing Deng, Xiaohui Shen, Daniel Cremers, and Liang-Chieh Chen.
\newblock An image is worth 32 tokens for reconstruction and generation, 2024.

\bibitem[Zeghidour et~al.(2021)Zeghidour, Luebs, Omran, Skoglund, and Tagliasacchi]{zeghidour2021soundstreamendtoendneuralaudio}
Neil Zeghidour, Alejandro Luebs, Ahmed Omran, Jan Skoglund, and Marco Tagliasacchi.
\newblock Soundstream: An end-to-end neural audio codec, 2021.

\bibitem[Zhang et~al.(2020)Zhang, Wang, Gan, Wu, Tenenbaum, Torralba, and Freeman]{Zhang2020Deep}
Zhoutong Zhang, Yunyun Wang, Chuang Gan, Jiajun Wu, Joshua~B. Tenenbaum, Antonio Torralba, and William~T. Freeman.
\newblock Deep audio priors emerge from harmonic convolutional networks.
\newblock In \emph{International Conference on Learning Representations}, 2020.
\newblock URL \url{https://openreview.net/forum?id=rygjHxrYDB}.

\bibitem[Zhao et~al.(2025)Zhao, Xiang, Ye, Li, Tian, Chen, Ding, and Wan]{zhao2025longcataudiocodecaudiotokenizerdetokenizer}
Xiaohan Zhao, Hongyu Xiang, Shengze Ye, Song Li, Zhengkun Tian, Guanyu Chen, Ke~Ding, and Guanglu Wan.
\newblock Longcat-audio-codec: An audio tokenizer and detokenizer solution designed for speech large language models, 2025.
\newblock URL \url{https://arxiv.org/abs/2510.15227}.

\end{thebibliography}
\bibliographystyle{iclr2026_conference}

\appendix
\section{Appendix}
\subsection{USE OF LLMs}
We only used large language models (LLMs) as a tool for language refinement and editing. They were not involved in the design of the methodology, experimental setup, data analysis, or any other core aspect of this research.

\subsection{ETHICS STATEMENT}
This study has been approved by the relevant ethics committee or institutional review board and was conducted in strict accordance with ethical guidelines. The rights, privacy, and welfare of participants were fully respected and protected, and all personal information was kept confidential.

Informed Consent: All participants were informed of the study’s objectives, procedures, potential risks, and benefits, either verbally or in writing, and provided their informed consent.

Data Confidentiality and Privacy Protection: Measures were implemented to safeguard participants’ personal information and ensure privacy.
\subsection{REPRODUCIBILITY STATEMENT}
Use of Research Data: All research data were collected, stored, and used in accordance with legal and ethical standards, ensuring transparency and proper interpretation.

We ensure that our method is fully reproducible. Upon acceptance of this paper, we will publicly release all code, model weights, and training data necessary to replicate our results.

\subsection{Equivalence of Spectrogram VGG Loss and Vocoder Feature Matching Loss}

Let $\hat{x}$ denote the waveform generated by a vocoder, and $x$ the corresponding ground-truth waveform.  
Define the Short-Time Fourier Transform (STFT) of the waveform as:
\begin{equation}
\hat{S} = \mathrm{STFT}(\hat{x}), \quad S = \mathrm{STFT}(x)
\end{equation}

\subsubsection{Spectrogram VGG Loss}

A spectrogram-based VGG loss is defined as the L1 distance between feature maps extracted from a convolutional network $\phi$ (e.g., VGG) applied to the spectrograms:
\begin{equation}
\mathcal{L}_{\mathrm{VGG}}(\hat{S}, S) = \sum_{l=1}^{L} w_l \, \| \phi_l(\hat{S}) - \phi_l(S) \|_1
\end{equation}
where $\phi_l(\cdot)$ is the feature map at the $l$-th layer, $w_l$ is a weighting coefficient, and $L$ is the total number of layers considered.

\subsubsection{Vocoder Feature Matching Loss}

In vocoder GANs, the feature matching loss is defined using the discriminator $D$:
\begin{equation}
\mathcal{L}_{\mathrm{FM}}(\hat{x}, x) = \sum_{d \in \mathcal{M}} \sum_{l=1}^{L_d} \| D_l^{(d)}(\hat{x}) - D_l^{(d)}(x) \|_1
\end{equation}
where $D_l^{(d)}(\cdot)$ denotes the feature map of the $l$-th layer of the $d$-th discriminator, and $\mathcal{D}$ is the set of discriminators (e.g., multi-resolution discriminators).  

Each discriminator first computes a spectrogram of the waveform:
\begin{equation}
X = \mathrm{STFT}(\cdot)
\end{equation}
and then applies a sequence of convolutional layers with non-linearities:
\begin{equation}
D_l^{(d)}(\hat{x}) = \sigma(W_l^{(d)} * X + b_l^{(d)}),
\end{equation}
where $W_l^{(d)}, b_l^{(d)}$ are the convolutional weights and biases, and $\sigma(\cdot)$ is the activation function.

\subsubsection{Equivalence}

Substituting $X = \mathrm{STFT}(\hat{x})$ and $X = \mathrm{STFT}(x)$ into the feature matching loss, we obtain:
\begin{equation}
\mathcal{L}_{\mathrm{FM}}(\hat{x}, x) = \sum_{d,l} \Big\| \sigma(W_l^{(d)} * \mathrm{STFT}(\hat{x}) + b_l^{(d)}) - \sigma(W_l^{(d)} * \mathrm{STFT}(x) + b_l^{(d)}) \Big\|_1
\end{equation}

Comparing with the spectrogram VGG loss in Eq.~(2), we see that the two losses share the same mathematical form:
\begin{equation}
\sum_l \| F_l(\mathrm{STFT}(\hat{x})) - F_l(\mathrm{STFT}(x)) \|_1
\end{equation}
where $F_l$ denotes a convolutional feature extractor. The only difference lies in the choice of network parameters (pretrained VGG weights vs. learned discriminator weights).  

\subsection{Equivalence of VGG Loss and Feature-Matching Loss}
\label{equivalence}
Both the spectrogram-based VGG loss and the vocoder feature matching loss are equivalent in the sense that they compute an L1 distance in the convolutional feature space of a spectrogram. Formally,
\begin{equation}
\mathcal{L}_{\mathrm{VGG}}(\hat{S}, S) \approx \mathcal{L}_{\mathrm{FM}}(\hat{x}, x),
\end{equation}
up to the network weights. This shows that the vocoder feature matching loss can be interpreted as a generalized, learnable spectrogram-based perceptual loss.

\subsection{Proof of ISTFT Lipschitz Continuity}
\label{app:istft_lipschitz}

We prove that the inverse short-time Fourier transform (ISTFT) operator is Lipschitz continuous with respect to the spectrogram input. 

\begin{lemma}[ISTFT Lipschitz Continuity]
Let $\mathcal{S}\in \mathbb{C}^{F\times T}$ be a complex spectrogram obtained by short-time Fourier transform (STFT) with analysis window $g\in\mathbb{R}^N$ and hop size $H$. Define the ISTFT operator $\mathrm{ISTFT}:\mathbb{C}^{F\times T}\to \mathbb{R}^M$ with synthesis window $h$. Then, for any two spectrograms $\mathcal{S}_1,\mathcal{S}_2$,
\[
\|\mathrm{ISTFT}(\mathcal{S}_1) - \mathrm{ISTFT}(\mathcal{S}_2)\|_2 
\;\le\; L_{\mathrm{ISTFT}} \,\|\mathcal{S}_1 - \mathcal{S}_2\|_2,
\]
where the Lipschitz constant $L_{\mathrm{ISTFT}}$ depends only on the window functions and hop size.
\end{lemma}

\begin{proof}
Recall that ISTFT reconstructs the waveform by overlap-add (OLA) of inverse FFTs of each frame:
\[
\hat{x}[n] = \sum_{t} h[n-tH] \cdot \mathrm{IFFT}(\mathcal{S}[:,t])[n-tH].
\]
Let $\Delta \mathcal{S} = \mathcal{S}_1 - \mathcal{S}_2$. Then the waveform difference is
\[
\Delta x[n] = \sum_{t} h[n-tH] \cdot \mathrm{IFFT}(\Delta \mathcal{S}[:,t])[n-tH].
\]

By Parseval's theorem, the $\ell_2$ norm of the IFFT is equal to $\ell_2$ norm  of the spectrum:
\[
\|\mathrm{IFFT}(\Delta \mathcal{S}[:,t])\|_2 
= \sqrt{N}\, \|\Delta \mathcal{S}[:,t]\|_2,
\]
where $N$ is the FFT length.

Applying Minkowski’s inequality to the OLA sum:
\[
\|\Delta x\|_2
\le \sum_{t} \|h(\cdot - tH)\|_\infty \cdot \|\mathrm{IFFT}(\Delta \mathcal{S}[:,t])\|_2.
\]

Since the shifted window has the same maximum magnitude as $h$,
\[
\|\Delta x\|_2 \le \|h\|_\infty \cdot \sqrt{N}\, \sum_{t} \|\Delta \mathcal{S}[:,t]\|_2.
\]

Finally, by Cauchy–Schwarz,
\[
\sum_{t} \|\Delta \mathcal{S}[:,t]\|_2
\le \sqrt{T} \,\|\Delta \mathcal{S}\|_2.
\]

Combining the inequalities, we obtain
\[
\|\mathrm{ISTFT}(\mathcal{S}_1) - \mathrm{ISTFT}(\mathcal{S}_2)\|_2
\;\le\; \|h\|_\infty \cdot \sqrt{NT}\,\|\mathcal{S}_1 - \mathcal{S}_2\|_2.
\]

Thus ISTFT is Lipschitz continuous with constant
\[
L_{\mathrm{ISTFT}} \le \|h\|_\infty \cdot \sqrt{NT}.
\]
\end{proof}

\begin{remark}
In practice, when $h$ is chosen as the canonical synthesis window satisfying the perfect-reconstruction condition (e.g., Hann window with $50\%$ overlap), $\|h\|_\infty \le 1$. Hence the ISTFT operator has a moderate Lipschitz constant that scales with the FFT length $N$ and number of frames $T$, ensuring stability against spectrogram perturbations.
\end{remark}

\subsection{Architecture Details For Tokenization}
\label{app:architecture}

\begin{table}[h]
    \centering
    \caption{The encoder architecture details of MelTok}
    \label{tab:placeholder_label}
    \begin{tabular}{llll}
        \toprule
        Layer Name & \makecell{Input Shape \\ (Batch,Channel,Frequency,Time)} & Kernel Size, Stride & Output Shape \\
        \midrule
        Patcher & (8, 1, 96, 96) &  & (8, 1, 48, 48) \\
        PreConv2D & (8, 1, 48, 48) & (3, 1) & (8, 128, 48, 48) \\
        \midrule
        \multicolumn{4}{l}{Downsampling 1:} \\
        \quad ResnetBlock & (8, 128, 48, 48) & (3, 1), (1, 1) & (8, 256, 48, 48) \\
        \quad ResnetBlock & (8, 256, 48, 48) & (3, 1), (1, 1) & (8, 256, 48, 48) \\
        \quad DownConv2D & (8, 256, 48, 48) & (2, 2) & (8, 256, 24, 24) \\
        \quad AttnBlock (Cond.) & (8, 256, 24, 24) & (1, 1) & (8, 256, 24, 24) \\
        \midrule
        \multicolumn{4}{l}{Downsampling 2:} \\
        \quad ResnetBlock & (8, 256, 24, 24) & (3, 1), (1, 1) & (8, 512, 24, 24) \\
        \quad ResnetBlock & (8, 512, 24, 24) & (3, 1), (1, 1) & (8, 512, 24, 24) \\
        \quad DownConv2D & (8, 512, 24, 24) & (2, 2) & (8, 512, 12, 12) \\
        \quad AttnBlock (Cond.) & (8, 512, 12, 12) & (1, 1) & (8, 512, 12, 12) \\
        ResnetBlock $\times$ 2 & (8, 512, 12, 12) & (3, 1), (1, 1) & (8, 512, 12, 12) \\
        AttnBlock (Cond.) & (8, 512, 12, 12) & (1, 1) & (8, 512, 12, 12) \\
        Normalizer & (8, 512, 12, 12) &  & (8, 512, 12, 12) \\
        PostConv2D & (8, 512, 12, 12) & (3, 1) & (8, 512, 12, 12) \\
        \bottomrule
    \end{tabular}
    \label{Table:encoder_detail}
\end{table}

\begin{table}[h]
    \centering
    \caption{2D Tokenizer vs 1D Tokenizer Architecture in Ablation Study}
    \label{tab:tokenizer_comparison}
    \begin{tabular}{>{\bfseries}l l l}
        \toprule
        \bfseries & \bfseries 2D Tokenizer & \bfseries 1D Tokenizer \\
        \midrule
        Input & $96 \times 96$ & $96 \times 96$ \\
        \hline
        Encoder & $96 \times 96 \to 4 \times 48 \times 48$ (Haar, patch=2) & $96 \times 96 \to 128 \times 96$ \\
                & $4 \times 48 \times 48 \to 128 \times 48 \times 48$ & $128\times 96$ → $256\times96$ \\
                & Downsample: & Downsample: \\
                & $128 \times 48 \times 48 \to 256 \times 24 \times 24 \to 512 \times 12 \times 12$ & $256 \times 96 \to 512 \times 12$ \\
        \hline
        Pre quant conv & $512 \times 12 \times 12 \to 256 \times 12 \times 12$ & $512 \times 12 \to 256 \times 12$ \\
        \hline
        Quant\_conv & $256 \times 12 \times 12 \to 64 \times 12 \times 12$ & $256 \times 12 \to 64 \times 12$ \\
        \hline
        Quantizer & 8192 codes of dim 64 & 8192 codes of dim 64 \\
                  & One Quantizer & 12 Quantizers to keep same bitate \\
        \hline
        Post quant conv & $64 \times 12 \times 12 \to 256 \times 12 \times 12$ & $64 \times 12 \to 256 \times 12$ \\
        \hline
        Decoder & Symmetric to Encoder & Symmetric to Encoder \\
        \hline
        Reconstruction & $96 \times 96$ & $96 \times 96$ \\
        \bottomrule
    \end{tabular}
    \label{tabel: ablation_study_conv}
\end{table}

\end{document}